\documentclass[conference,a4paper]{IEEEtran}
\usepackage{amsfonts}
\usepackage[sort]{cite} 
\usepackage{amsmath,amsthm} 
\usepackage{amssymb}  
\usepackage{graphicx}
\usepackage{latexsym}
\usepackage{array}
\usepackage{amsfonts}
\usepackage{bm}
\usepackage{mathtools}
\usepackage{etoolbox}
\usepackage{lipsum}
\usepackage{caption}
\usepackage{subcaption}

\def\IID{\mathop{\mathrm{i.i.d.}}}  
\def\diag{\mathop{\mathrm{diag}}}   
\def\trace{\mathop{\mathrm{trace}}} 
\def\op{\mathop{\mathrm{op}}}       
\def\rdf{\mathop{\mathrm{RDF}}}     
\def\awgn{\mathop{\mathrm{AWGN}}}   
\def\na{\mathop{\mathrm{na}}}       
\def\mse{\mathop{\mathrm{MSE}}}     
\def\mimo{\mathop{\mathrm{MIMO}}}     
\def\nrdf{\mathop{\mathrm{NRDF}}}   
\def\opta{\mathop{\mathrm{OPTA}}}   
\def\op{\mathop{\mathrm{op}}}       
\def\zd{\mathop{\mathrm{ZD}}}       
\def\sd{\mathop{\mathrm{SD}}}       
\def\sdusq{\mathop{\mathrm{SDUSQ}}} 
\def\ecdq{\mathop{\mathrm{ECDQ}}} 
\def\p1{\mathop{\mathrm{P1}}}       
\def\p2{\mathop{\mathrm{P2}}}       
\def\unif{\mathop{\mathrm{Unif}}}     
\def\rvs{\mathop{\mathrm{RVs}}}     
\def\rv{\mathop{\mathrm{RV}}}       
\def\ar{\mathop{\mathrm{AR}}}       

\newtheorem{theorem}{Theorem}

\newtheorem{definition}{Definition}
\newtheorem{corollary}{Corollary}

\newtheorem{remark}{Remark}

\newtheorem{example}{Example}

\newcommand{\T}{^{\mbox{\tiny T}}}

\newcommand{\noi}{\noindent}

\newcommand{\be}{\begin{equation}}
\newcommand{\ee}{\end{equation}}
\newcommand{\bea}{\begin{eqnarray}}
\newcommand{\eea}{\end{eqnarray}}
\newcommand{\bes}{\begin{eqnarray*}}
\newcommand{\ees}{\end{eqnarray*}}

\newcommand{\bfi}{\begin{figure}}
\newcommand{\bfit}{\begin{figure}[t]}
\newcommand{\bfib}{\begin{figure}[b]}
\newcommand{\bfih}{\begin{figure}[h]}
\newcommand{\bfip}{\begin{figure}[p]}
\newcommand{\efi}{\end{figure}}

\newcommand{\bi}{\begin{itemize}}
\newcommand{\ei}{\end{itemize}}
\newcommand{\ben}{\begin{enumerate}}
\newcommand{\een}{\end{enumerate}}

\newcommand{\bp}{\begin{problem}}
\newcommand{\ep}{\end{problem}}

\begin{document}

\sloppy

\title{An Upper Bound to Zero-Delay Rate Distortion via Kalman Filtering for Vector Gaussian Sources}

\author{
  \IEEEauthorblockN{Photios A. Stavrou$^*$, Jan $\O$stergaard$^*$, Charalambos D. Charalambous$^\dagger$, Milan Derpich$^\ddagger$}
  \IEEEauthorblockA{$^*$Department of Electronic Systems, Aalborg University, Denmark\\
  $^\dagger$Department of Electrical and Computer Engineering, University of Cyprus, Cyprus\\
  $^\ddagger$ Department of Electronic Engineering,  Universidad T\'ecnica Federico Santa Mar\'ia, Chile.\\
 {\it Emails:\{fos, jo\}@es.aau.dk, chadcha@uc.ac.cy, milan.derpich@usm.cl}}}

\maketitle

\begin{abstract}
We deal with zero-delay source coding of a vector Gaussian autoregressive ($\ar$) source subject to an average mean squared error ($\mse$) fidelity criterion. Toward this end, we consider the nonanticipative rate distortion function ($\nrdf$) which is a lower bound to the causal and zero-delay rate distortion function ($\rdf$). We use the realization scheme with feedback proposed in \cite{stavrou-charalambous-charalambous2016} to model the corresponding optimal ``test-channel'' of the $\nrdf$, when considering vector Gaussian $\ar (1)$ sources subject to an average $\mse$ distortion. We give conditions on the vector Gaussian $\ar (1)$ source to ensure asymptotic stationarity of the realization scheme (bounded performance). Then, we encode the vector innovations due to Kalman filtering via lattice quantization with subtractive dither and memoryless entropy coding. This coding scheme provides a tight upper bound to the zero-delay Gaussian $\rdf$. We extend this result to vector Gaussian $\ar$ sources of any finite order. Further, we show that for infinite dimensional vector Gaussian $\ar$ sources of any finite order, the $\nrdf$ coincides with the zero-delay $\rdf$. Our theoretical framework is corroborated with a simulation example.
\end{abstract}

%
%
%
%
%
\section{Introduction}\label{sec:introduction}

\par Zero-delay source coding is desirable in various real-time applications, such as, in signal processing \cite{huang-benesty2004} and networked control systems \cite{linder:2014,tanaka-johansson-oechtering-sandberg-skoglund2016,mohsen:2017}. Zero-delay codes form a subclass of causal source codes  (see \cite{neuhoff-gilbert1982}), namely, codes where the reproduced source samples depends on the source samples in a causal manner. However, zero-delay source coding compared to causal source coding allow the reproduction of each source sample at the same time instant that the source sample is encoded. Unfortunately, causal source coding does not exclude the possibility of long blocks of quantized samples, which may cost arbitrary end-to-end delays. 
\par Zero-delay codes (and causal codes) in constrast to non-causal codes cannot achieve the classical rate distortion function ($\rdf$). Indeed, an open problem in information theory is quantifying the gap between the optimal performance theoretically attainable ($\opta$) by non-causal codes, and the $\opta$ by causal and zero-delay codes, hereinafter denoted by  $R_c^{\op}(D)$ and $R^{\op}_{\zd}(D)$, respectively. Notable exceptions where this gap is explicitly found are memoryless sources \cite{neuhoff-gilbert1982}, stationary sources in high rates \cite{linder-zamir2006}, and zero mean stationary scalar Gaussian sources with average mean squared error ($\mse$) distortion \cite{derpich-ostergaard2012}. 
\par Throughout the years, the interest in zero-delay applications is growing, thus, initiating further research on characterizing the fundamental limitations of the $\opta$ by zero-delay codes. Unfortunately, it turns out that $R^{\op}_{\zd}(D)$ is very hard to compute and for this reason there has been a turn in studying variants of classical $\rdf$ that perform as tight as possible to $R^{\op}_{\zd}(D)$. 
\par In this paper, we derive a tight upper bound to zero-delay source coding for vector Gaussian $\ar$ sources subject to an average $\mse$ distortion. We consider nonanticipative rate distortion function ($\nrdf$) (see, e.g., \cite{gorbunov-pinsker1972a,derpich-ostergaard2012,stavrou-charalambous-charalambous2016}), which gives a tighter lower bound to $R^{\op}_{\zd}(D)$ compared to the classical $\rdf$ (see, e.g., \cite[eq. (11)]{derpich-ostergaard2012}). Then, we employ the feedback realization scheme proposed in \cite[Fig. IV.3]{stavrou-charalambous-charalambous2016}, that corresponds to the optimal ''test-channel'' of $\nrdf$ for vector Gaussian $\ar (1)$ sources and average $\mse$ distortion. Further, we give conditions to asymptotically stabilize the performance of the specific scheme. By invoking standard techniques using entropy coded dithered quantizer ($\ecdq$) \cite{zamir-feder1992,zamir-feder1996} on the innovations' encoder of the feedback realization scheme, we derive the tight upper bound. In addition, we show how to generalize our scheme to vector Gaussian $\ar$ sources of any finite order. If the vector dimension of the Gauss $\ar$ source tends to infinity, we show that the $R^{\op}_{\zd}(D)$ coincides with the $\nrdf$. We demonstrate our results with a numerical example.

%
%
%
%
%
\paragraph*{Notation} We let $\mathbb{R}=(-\infty,\infty)$, $\mathbb{N}_0=\{0,1,\ldots\}$. For $t\in\mathbb{N}_0$. We denote a sequence of $\rvs$ by ${\bf x}^n\triangleq({\bf x}_0,\ldots,{\bf x}_n)$ and its realization by ${\bf x}^n={x}^n, x_j\in{\cal X}_j,~j=0,\ldots,n$. The distribution of the $\rv$ ${\bf x}$ on ${\cal X}$ is denoted by ${\bf P}_{\bf x}(dx)\equiv{\bf P}(dx)$. The conditional distribution of ${\rv}$ ${\bf y}$ given ${\bf x}=x$ is denoted by ${\bf P}_{{\bf y}|{\bf x}}(dy|{\bf x}=x)\equiv{\bf P}(dy|x)$. The transpose of a matrix ${S}$ is denoted by ${S}^{\T}$. For a square matrix $S\in \mathbb{R}^{p\times p}$ with entries $S_{ij}$ on the $i^{th}$ row and $j^{th}$ column, we denote by $\diag\{S\}$ the matrix having $S_{ii},~i=1,\ldots,p$, on its diagonal and zero elsewhere. 

%
%
%
%
%
\section{Problem Statement}\label{sec:problem_formulation}

\par In this paper we consider the zero-delay source coding setting illustrated in Fig.~\ref{fig:zero_delay_system}. In this setting, the $p-$dimensional (vector) Gaussian source is governed by the following discrete-time linear time-invariant state-space model
\begin{align}
{\bf x}_{t+1}=A{\bf x}_{t}+B{\bf w}_t,~t\in\mathbb{N}_0,\label{state_space_model}
\end{align}
\noi where $A\in\mathbb{R}^{p\times{p}}$, and $B\in\mathbb{R}^{p\times{q}}$ are known, ${\bf x}_0\in\mathbb{R}^p\sim{N}(0;\Sigma_{{\bf x}_0})$ is the initial state, and the noise process ${\bf w}_t\in\mathbb{R}^q$ is an $\IID$ Gaussian $N(0;I_{q\times{q}})$ sequence, independent of ${\bf x}_0$. We allow $A$ to have eigenvalues outside the unit circle which means that ${\bf x}_t$ can be unstable. 
 
\par The system operates as follows. At every time step $t$, the {\it encoder} observes the source ${\bf x}^t$ and produces a single binary codeword ${\bf z}_t$ from a predefined set of codewords ${\cal Z}_t$ of at most a countable number of codewords. Since the source is random, ${\bf z}_t$ and its length ${\bf l}_t$ are random variables. Upon receiving ${\bf z}_t$, the {\it decoder} produces an estimate ${\bf y}_t$ of the source sample. We assume that both the encoder and decoder process information without delay and they are allowed to have infinite memory of the past.
 
\begin{figure}[htp]
\centering
\includegraphics[width=90mm]{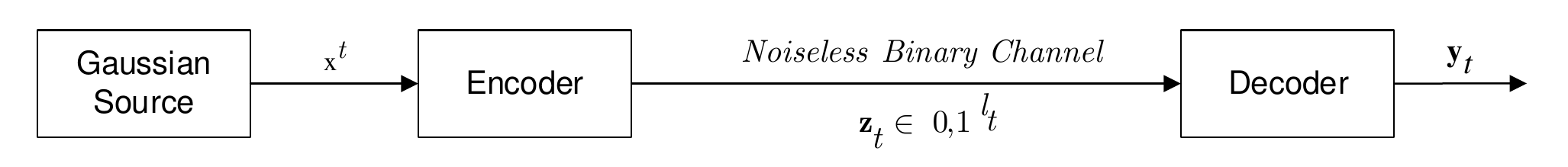}
\caption{A zero-delay source coding scenario.}\label{fig:zero_delay_system}
\end{figure}

\par The analysis of the noiseless digital channel is restricted to the class of instantaneous variable-length binary codes ${\bf z}_t$. The countable set of all codewords (codebook) ${\cal Z}_t$ is time-varying to allow the binary representation ${\bf z}_t$ to be an arbitrarily long sequence. The encoding and decoding policies are described by sequences of conditional probability distributions as $\{{\bf P}(dz_t|z^{t-1},x^t):~t\in\mathbb{N}_0\}$ and $\{{\bf P}(dy_t|y^{t-1},z^t):~t\in\mathbb{N}_0\}$, respectively. At $t=0$, we assume ${\bf P}(dz_0|z^{-1},x^0)={\bf P}(dz_0|x_0)$ and ${\bf P}(dy_0|y^{-1},z^0)={\bf P}(dy_0|z_0)$.
\par The design in Fig.~\ref{fig:zero_delay_system} is required to yield an asymptotic average distortion $\limsup_{n\longrightarrow\infty}\frac{1}{n+1}\mathbb{E}\{d({\bf x}^n,{\bf y}^n)\}\leq{D}$, where $D>0$ is the pre-specified distortion level, $d({\bf x}^n,{\bf y}^n)\triangleq\sum_{t=0}^n||{\bf x}_t-{\bf y}_t||_{2}^2$. The objective is to minimize the expected average codeword length denoted by $\limsup_{n\longrightarrow\infty}\frac{1}{n+1}\sum_{t=0}^n\mathbb{E}({\bf l}_t)$, over all encoding-decoding policies. These design requirements are formally cast by the following optimization problem:
\begin{align}
R^{\op}_{\zd}(D)\triangleq&\inf\limsup_{n\longrightarrow\infty}\frac{1}{n+1}\sum_{t=0}^n\mathbb{E}({\bf l}_t)\label{def:operational_zero_delay}\\
&\mbox{s. t.}~~\limsup_{n\longrightarrow\infty}\frac{1}{n+1}\mathbb{E}\{d({\bf x}^n,{\bf y}^n)\}\leq{D},\nonumber
\end{align}
\noindent i.e., the $\opta$ by zero-delay codes.  

%
%
%

\section{Preliminaries}\label{sec:preliminaries}

\par In this section, we give the definition of $\nrdf$ of vector Gaussian $\ar$ sources subject to an average $\mse$ distortion.  
\par Define the source distribution by ${\bf P}(dx^n)\triangleq\prod_{t=0}^n{\bf P}(dx_t|{x}^{t-1})$, the reconstruction distribution by ${\bf P}(dy^n||x^n)\triangleq\prod_{t=0}^n{\bf P}(dy_t|y^{t-1},x^t)$, and the joint distribution by ${\bf P}(d{x}^n,d{y}^n)\triangleq{\bf P}(dx^n)\otimes{\bf P}(dy^n||x^n)$. The marginal on ${y}_t\in{\cal Y}_t$, ${\bf P}(d{y}_t|{y}^{t-1})$, is induced by the joint distribution ${\bf P}(d{x}^n,d{y}^n)$. We assume that at $t=0$, ${\bf P}(dy_0|y^{-1},x^0)={\bf P}(dy_0|x_0)$.
 
\par Given the previous distributions, we introduce the mutual information between ${\bf x}^n$ and ${\bf y}^n$ as follows
\begin{align*}
I({\bf x}^n;{\bf y}^n)&\triangleq=\sum_{t=0}^n\mathbb{E}\log(\frac{{\bf P}({\bf y}_t|{\bf y}^{t-1},{\bf x}^t)}{{\bf P}({\bf y}_t|{\bf y}^{t-1})}),
\end{align*}
where $\mathbb{E}\{\cdot\}$ is the expectation with respect to the joint distribution ${\bf P}(d{x}^n,d{y}^n)$.
\begin{definition}($\nrdf$ with average $\mse$ distortion)\label{nonanticipative_rdf}{\ \\}
(1) The finite-time $\nrdf$ is defined by 
\begin{align}
{R}^{\na}_{0,n}(D) \triangleq & \inf_{\substack{{\bf P}(dy_t|y^{t-1},x^t):~t=0,\ldots,n\\~\frac{1}{n+1}\mathbb{E}\{d({\bf x}^n,{\bf y}^n)\}\leq{D}}}\frac{1}{n+1} I({\bf x}^n;{\bf y}^n),\label{finite_time_nrdf}
\end{align}
assuming the infimum exists.\\
(2) The per unit time asymptotic limit of \eqref{finite_time_nrdf} is defined by 
\begin{align}
{R}^{\na}&(D)=\lim_{n\longrightarrow\infty}{R}^{\na}_{0,n}(D),\label{infinite_time_nrdf}
\end{align}
assuming the infimum exists and the limit exists and it is finite.
\end{definition}
\par If one replaces $\liminf$ by $\inf\lim$ in \eqref{infinite_time_nrdf}, then an upper bound to $R^{\na}(D)$ is obtained, defined as follows.
\begin{align}
\bar{R}^{\na}(D)\triangleq&\inf_{{\bf P}^{\infty}(dy|z,x)}\lim_{n\longrightarrow\infty}\frac{1}{n+1}{R}^{\na}_{0,n}(D),\label{infinite_time_nrdf_upper_bound}
\end{align}
where ${\bf P}^{\infty}(dy|z,x)$ is the stationary or time-invariant reconstruction distribution.\\
It is shown in \cite[Theorem 6.6]{stavrou2016} that provided the limit in \eqref{infinite_time_nrdf_upper_bound} exists, and the source is stationary (or asymptotically stationary) then ${R}^{\na}(D)=\bar{R}^{\na}(D)$.
\par The optimization problem of Definition \ref{nonanticipative_rdf}, in contrast to the one given in \eqref{def:operational_zero_delay} is convex (see e.g., \cite{charalambous-stavrou2016}). In addition, for the source model \eqref{state_space_model} and the average $\mse$ distortion, then, by \cite[Theorems 1]{stavrou-charalambous-charalambous2016}, the optimal ``test channel'' corresponding to \eqref{infinite_time_nrdf} is of the form  
\begin{align}
{\bf P}^*(dy_t|y^{t-1},x^t)={\bf P}^*(dy_t|y^{t-1},x_t),~t\in\mathbb{N}_0,\label{optimal_minimizer}
\end{align}
where at $t=0$, ${\bf P}^*(dy_0|y_{-1},x_0)={\bf P}^*(dy_0|x_0)$, and the corresponding joint process $\{({\bf x}_t, {\bf y}_t):~t\in\mathbb{N}_0\}$  is jointly Gaussian.  

\section{Asymptotically Stationary Feedback Realization Scheme via Kalman Filtering}\label{subsec:realization_scheme}

\begin{figure*}[htp]
\centering
\includegraphics[width=0.9\textwidth]{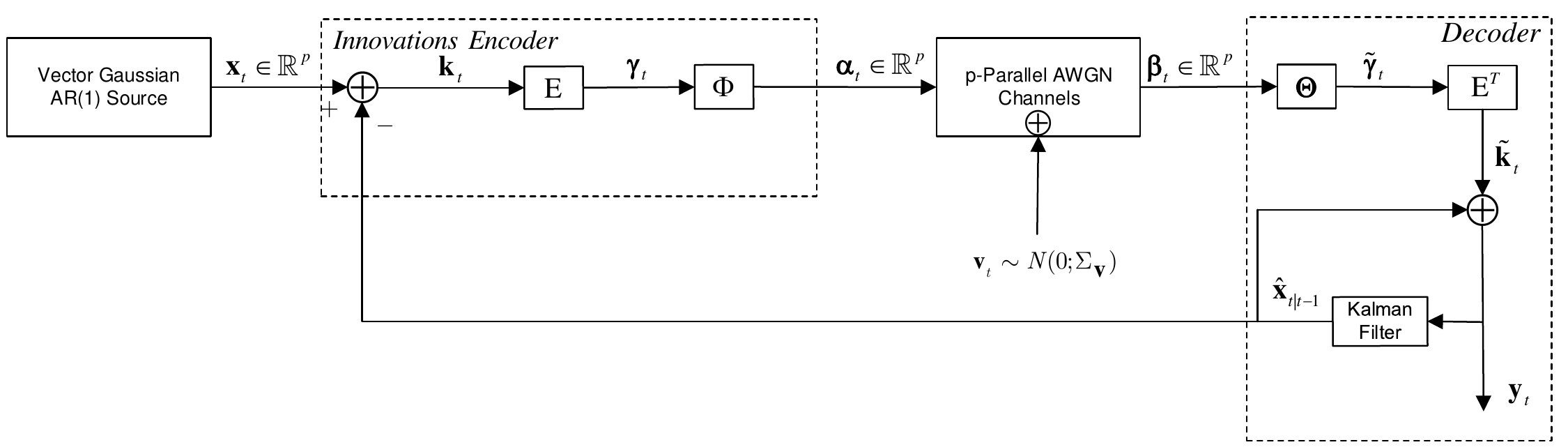}
\caption{Asymptotically stationary feedback realization scheme corresponding to \eqref{optimal_realization}.}
\label{fig:noisy_communication_system}
\end{figure*}

\par The authors in \cite[Theorem 2]{stavrou-charalambous-charalambous2016} realized the optimal ``test channel'' of \eqref{optimal_minimizer} with the feedback realization scheme illustrated in \cite[Fig. IV.3]{stavrou-charalambous-charalambous2016} that corresponds to a realization of the form:
\begin{align}
{\bf y}_t=E_t^{\T}{H}_t{E}_t({\bf x}_t-\widehat{\bf x}_{t|t-1})+E_t^{\T}{\Theta}_t{\bf v}_t+\widehat{\bf x}_{t|t-1},\label{optimal_realization}
\end{align}
where ${H}_t\triangleq\Phi_t\Theta_t$ is a scaling matrix; ${\bf v}_t$ is an independent Gaussian noise process with $N(0;\Sigma_{{\bf v}_t})$,~$\Sigma_{{\bf v}_t}=\diag\{V_t\}$ independent of ${\bf x}_0$; the error ${\bf x}_t-\widehat{\bf x}_{t|t-1}$ is Gaussian with ${N}(0; \Pi_{t|t-1})$, and $\widehat{\bf x}_{t|t-1}\triangleq\mathbb{E}\{{\bf x}_t|{\bf y}^{t-1}\}$; the error ${\bf x}_t-\widehat{\bf x}_{t|t}$ is Gaussian with ${N}(0; \Pi_{t|t})$, and $\widehat{\bf x}_{t|t}\triangleq\mathbb{E}\{{\bf x}_t|{\bf y}^{t}\}$. Moreover, $\{\widehat{\bf x}_{t|t-1},~\Pi_{t|t-1}\}$ are given by the following Kalman filter equations:
\begin{subequations}\label{kalman_filter_finite_time}
\begin{align}
&\text{Prediction:}~\nonumber\\
&\widehat{\bf x}_{t|t-1}=A\widehat{\bf x}_{t-1|t-1},~~~\widehat{\bf x}_{0|-1}=\mathbb{E}\{{\bf x}_0\},\\
&\Pi_{t|t-1}={A}\Pi_{t-1|t-1}A^{\T}+BB^{\T},~~\Pi_{0|-1}=\Sigma_{{\bf x}_0},\label{a_priori_covariance}\\
&\text{Update:}~\nonumber\\
&\widehat{\bf x}_{t|t}=\widehat{\bf x}_{t|t-1}+G_t\tilde{\bf k}_t,~~~{\bf x}_{0|0}=\mathbb{E}\{{\bf x}_0\},\\
&\tilde{\bf k}_t\triangleq {\bf y}_t-\widehat{\bf x}_{t|t-1}~~\text{(innovation)},\label{innovations}\\
&\Pi_{t|t}=\Pi_{t|t-1}-G_tS_tG_t^{\T},~~\Pi_{0|0}=\Sigma_{{\bf x}_0},\label{a_posteriori_covariance}\\
&G_t=\Pi_{t|t-1}(E_t^{\T}H_tE_t)^{\T}S^{-1}_t,~~\text{(Kalman Gain)} \nonumber\\
&S_t=(E_t^{\T}H_tE_t)\Pi_{t|t-1}(E_t^{\T}H_tE_t)^{\T}+E_t^{\T}\Theta_t\Sigma_{{\bf v}_t}\Theta_t^{\T}E_t,\nonumber
\end{align}
\end{subequations}
where $H_t\triangleq\diag\left\{1-\frac{\tilde{\Delta}_{t}}{\Lambda_{t}}\right\}$,~${\Theta}_t\triangleq\sqrt{H_t\tilde{\Delta}_t\Sigma_{{\bf v}_t}^{-1}}$,~$\Phi_t\triangleq{\Theta}_t^{-1}H_t$, $\tilde{\Delta}_t\triangleq\diag\{\delta_t\}$,~$\Lambda_t\triangleq{E}_t\Pi_{t|t-1}E_t^{\T}=\diag\{\lambda_t\}$, and $E_t\in\mathbb{R}^{p\times{p}}$ is an orthogonal matrix. It is easy to verify following that the following hold:
\begin{align}
\widehat{\bf x}_{t|t-1}=A{\bf y}_{t-1},~\widehat{\bf x}_{t|t}={\bf y}_{t},~G_t=I,~\Pi_{t|t}=E_t^{\T}\tilde{\Delta}_tE_t,\label{new_estimates}
\end{align}
where $I\in\mathbb{R}^{p\times{p}}$ denotes the identity matrix. By substituting \eqref{new_estimates} in \eqref{optimal_realization} we can also deduce that ${\bf P}^*(dy_t|y^{t-1},x_t)={\bf P}^*(dy_t|y_{t-1},x_t)$. 
\par The realization scheme of \eqref{optimal_realization} becomes asymptotically stationary (stable) if one of the following two conditions hold:
(1) $A$ is stable, i.e., its eigenvalues have magnitude less than one; (2) the pair ($A,B$) is (completely) stabilizable (see e.g., \cite[p. 342]{anderson-moore:1979}). This means that $\Pi\triangleq\lim_{t\longrightarrow\infty}\Pi_{t|t-1}<\infty$, $\Pi^\prime\triangleq\lim_{t\longrightarrow\infty}\Pi_{t|t}<\infty$, $E_t\equiv{E}$, and $\Sigma_{{\bf v}_t}\equiv{\Sigma}_{\bf v}$. 
\par Next, we briefly discuss the resulting asymptotically stationary realization scheme depicted in Fig. \ref{fig:noisy_communication_system}.

\paragraph*{\it Preprocessing at Encoder} Introduce the estimation error $\{{\bf k}_t\in\mathbb{R}^p:~{t\in\mathbb{N}_0}\}$, where ${\bf k}_t\triangleq{\bf x}_t-\widehat{\bf x}_{t|t-1},~t\in\mathbb{N}_0$ with (error) covariance $\Pi_{t|t-1}$,~$t\in\mathbb{N}_0$. Under conditions (1) or (2), we ensure that $\Pi\triangleq\lim_{t\longrightarrow\infty}\Pi_{t|t-1}$ and it is unique. The error covariance matrix $\Pi$ is diagonalized by introducing an orthogonal matrix $E$ (invertible matrix) such that $E\Pi{E}^{\T}=\diag\{\lambda\}\triangleq{\Lambda}$. To facilitate the computation, we introduce the scaling process $\{{\bm \gamma}_t\in\mathbb{R}^p:~{t\in\mathbb{N}_0^{n}}\}$, where ${\bm \gamma}_t\triangleq{E}{\bf k}_t,~{t\in\mathbb{N}_0}$, has independent Gaussian components.
\paragraph*{\it Preprocessing at Decoder} Analogously, we introduce the innovations process $\{{\bf \tilde{k}}_t:~{t\in\mathbb{N}_0}\}$ defined by \eqref{innovations} and the scaling process  $\{{ \tilde{\bm \gamma}}_t:~{t\in\mathbb{N}_0}\}$ defined by ${\bm \tilde{\gamma}}_t{\triangleq}\Theta{\bm \beta}_t$,~with ${\bm \beta}_t\triangleq\left(\Phi{E}{\bf k}_t+{\bf v}_t\right),~~{\bf v}_t\sim{N}(0;\Sigma_{\bf v})$, and $\{{\Phi},~{\Theta}\}$ are the asymptotic limits of $\Phi_t$ and $\Theta_t$, respectively.

\noindent The fidelity criterion $||{\bf k}_t-{\bf \tilde{k}}_t||_{2}^2$ at each $t$ is not affected by the above processing of $\{({\bf x}_t,{\bf y}_t):~{t\in\mathbb{N}_0}\}$, in the sense that the preprocessing at both the encoder and decoder do not affect the form of the squared error distortion function, that is,
\begin{align}
||{\bf x}_t-{\bf y}_t||_{2}^2=||{\bf k}_t-{\bf \tilde{k}}_t||_{2}^2,~~t\in\mathbb{N}_0.\label{feedback_squeme:eq.3}
\end{align}
\noi Moreover, using basic properties of conditional entropy (see, e.g., \cite[Eq. (IV.35)]{charalambous-stavrou-ahmed2014}), it can be shown that
\begin{align}
&{R}^{\na}(D)=\lim\frac{1}{n+1}\sum_{t=0}^n{I}({\bf x}_t;{\bf y}_t|{\bf y}^{t-1})\label{feedback_squeme:eq.4}\\
&\mbox{ s.t.}~~\lim_{n\longrightarrow\infty}\frac{1}{n+1}\mathbb{E}\left\{d({\bf x}^n,{\bf y}^n)\right\}\leq{D},\nonumber
\end{align}
and
\begin{align}
&R^{\na,{\bf k}^n,{\bf \tilde{k}}^n}(D)\triangleq\lim\frac{1}{n+1}\sum_{t=0}^n{I}({\bf k}_t;{\bf \tilde{k}}_t)\label{feedback_squeme:eq.6}\\
&\mbox{ s.t.}~~\lim_{n\longrightarrow\infty}\frac{1}{n+1}\mathbb{E}\left\{d({\bf k}^n,{\bf \tilde{k}}^n)\right\}\leq{D},\nonumber
\end{align}
are equivalent expressions.\\
In addition, the steady state values of \eqref{a_priori_covariance} is $\Pi={A}\Pi^\prime{A}^{\T}+BB^{\T}$. The end-to-end $\mse$ distortion of the scheme in Fig. \ref{fig:noisy_communication_system} is
\begin{align}
&\lim_{t\longrightarrow\infty}\mathbb{E}\{({\bf x}_t-{\bf y}_t)^{\T}({\bf x}_t-{\bf y}_t)\}\nonumber\\
&=\lim_{t\longrightarrow\infty}\trace\mathbb{E}\{({\bf x}_t-\widehat{\bf x}_{t|t})({\bf x}_t-\widehat{\bf x}_{t|t})^{\T}\}=\trace(\Pi^{\prime})\leq{D}.\nonumber
\end{align}
\noi Hence, following \cite[Theorem 2]{stavrou-charalambous-charalambous2016}, the per unit time asymptotic limit of Gaussian $\nrdf$ subject to the total $\mse$ distortion can be expressed as follows.
\begin{align}
{R}^{\na}(D)&
=\min_{\trace(\Pi^{\prime})\leq{D}}\frac{1}{2}\log\max\left(1,\frac{|{A}\Pi^\prime{A}^{\T}+BB^{\T}|}{|\Pi^{\prime}|}\right).\label{optimal_solution:eq.1} 
\end{align}
\par Clearly, the optimization problem in \eqref{optimal_solution:eq.1}, is a log-determinant minimization problem and can be solved using, for instance, Karush-Kuhn-Tucker conditions \cite[Chapter 5.5.3]{boyd:2004} or semidefinite programming (SDP). A way of solving \eqref{optimal_solution:eq.1} is proposed in \cite{tanaka:2017}. However, compared to that work, our realization scheme is implemented with feedback to take into account the effect of unstable sources in the dynamical system
%
%
%
%
%
\section{Upper Bound to zero-delay Gaussian $\rdf$}
\label{sec:main_results}

\par In this section, we derive an upper bound to the zero-delay Gaussian $\rdf$ using a subtractively dithered uniform scalar quantizer ($\sdusq$) on the feedback realization scheme of Fig.~\ref{fig:noisy_communication_system}. The $\sdusq$ scheme was introduced in \cite{zamir-feder1992} and since then it has been used in several papers (see, e.g., \cite{silva-derpich-ostergaard2011,derpich-ostergaard2012,tanaka-johansson-oechtering-sandberg-skoglund2016}) under various realization setups. However, it has never been documented for the realization scheme proposed in this work. Here, we consider the vector Gaussian $\ar (1)$ source of \eqref{state_space_model}, and we quantize each time step $t$ over $p$ independently operating $\sdusq$, with their outputs being jointly entropy coded conditioned to the dither. We extend our results when using vector quantization showing that at infinite dimensional vectors, the space-filling loss due to compression and the entropy coding extinguishes, i.e., $R^{\na}(D)$ and $R^{\op}_{\zd}(D)$ coincide.

  
\subsection{Scalar quantization}\label{subsection:componentwise_quantization}  

\par Next, we use the asymptotically stationary feedback realization scheme illustrated in Fig. \ref{fig:noisy_communication_system} to design an efficient \texttt{\{encoder/quantizer,decoder\}} pair. 
\par We select the quantizer step size $\Delta$ so that the covariance of the resulting quantization error meets $\Sigma_{\bf v}$. The encoder does not quantize the observed state ${\bf x}_t$ directly. Instead, it quantizes the deviation of ${\bf x}_t$ from the linear estimate $\widehat{\bf x}_{t|t-1}$ of ${\bf x}_t$. This method is known in least squares estimation theory as {\it innovations approach} and, therefore, the encoder is named as an {\it innovations' encoder}. We consider the zero-delay source coding setup illustrated in Fig.~\ref{fig:noisy_communication_system} with the additional change of the $p-$parallel additive white Gaussian noise ($\awgn$) channels with $p$ independently operating $\sdusq$. This is illustrated in Fig.~\ref{fig:replacement}. Note that, all matrices and scalings adopted in Fig. \ref{fig:noisy_communication_system} still hold when the aforementioned replacement is applied.

\begin{figure}
    \centering
    \begin{subfigure}[b]{0.52\columnwidth}
        \includegraphics[width=\textwidth]{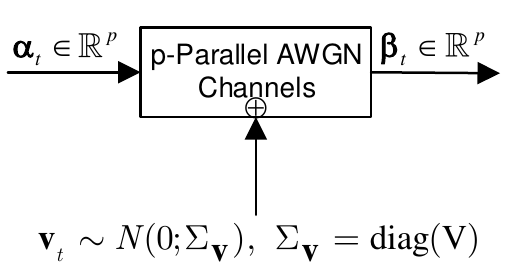}         
        \caption{$p-$parallel $\awgn$ channels.}
        \label{fig:awgn}
    \end{subfigure}
    ~ 
    \begin{subfigure}[b]{0.95\columnwidth}
        \includegraphics[width=\textwidth]{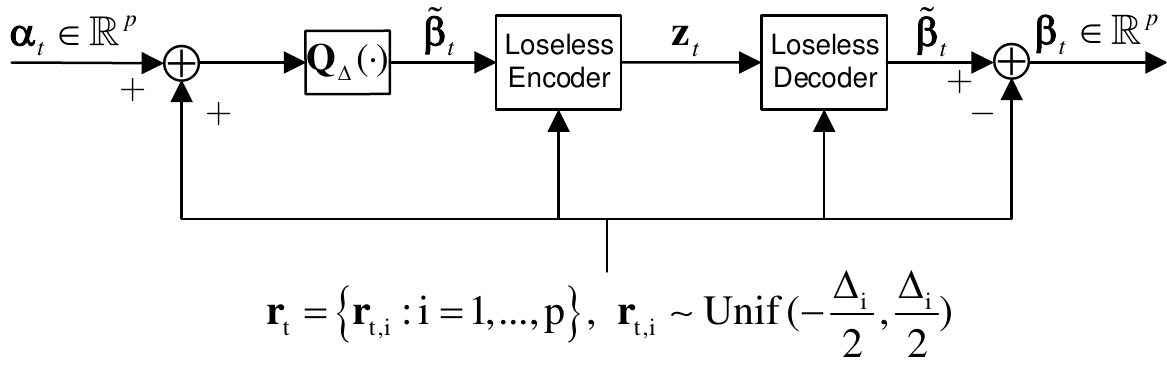}
        \caption{Realization over $p$ independently operating $\sdusq$.}
        \label{fig:sdusq}
    \end{subfigure}
    ~ 
    \caption{Scalar quantization by replacing a $p-$dimensional $\awgn$ channel with $p$ independently operating $\sdusq$.\vspace*{-0.3cm}}\label{fig:replacement}
\end{figure}

\par For each time step $t$, the input to the quantizer, is a scaled estimation error defined as follows
\begin{align}
{\bm{\alpha}}_t={\cal A}{\bf k}_t,~{\cal A}=\Phi{E}.\label{input_quantizer}
\end{align} 
Moreover, ${\bm \alpha}_t$ is an $\mathbb{R}^p-$valued random process. The parallel $p-$dimensional $\awgn$ channel is replaced by $p$ independently operating $\sdusq$, hence we can design the covariance matrix $\Sigma_{\bf v}$ of the $\awgn$ channels corresponding to the $p$-parallel $\awgn$ channels in such a way, that for each $t$, each diagonal entry $V_{ii}, i=1,\ldots,p$, i.e., $\Sigma_{\bf v}\triangleq\diag\{V\}$, to correspond to a quantization step size $\Delta_i, i=1,\ldots,p$, such that 
\begin{align}
{V}_{ii}=\frac{\Delta^2_i}{12},~i=1,\ldots,p.\label{quantization_step}
\end{align}
This results creates a multi-input multi-output ($\mimo$) transmission of parallel and independent $\sdusq$. We apply $\sdusq$ to each component of ${\bm \alpha}_t$, i.e.,
\begin{align}
{\bm \beta}_{t,i}=Q^{\sd}_{\Delta_i}({\bm \alpha}_{t,i}),~i=1,\ldots,p\label{output_quantizer1}
\end{align}
 and we let ${\bf r}_t$ be the $\mathbb{R}^p-$valued random process of dither signals whose individual components  $\{{\bf r}_{t,1},\ldots,{\bf r}_{t,p}\}$ are mutually independent and uniformly distributed random variables ${\bf r}_{t,i}\sim{\unif}\left(-\frac{\Delta_i}{2},\frac{\Delta_i}{2}\right)$ independent of the corresponding source input components ${\bm{\alpha}}_{t,i},~\forall{t,i}$. The output of the quantizer is given by
\begin{align}
\tilde{\bm \beta}_{t,i}=Q_{\Delta_i}({\bm \alpha}_{t,i}+{\bf r}_{t,i}),~i=1,\ldots,p.\label{output_quantizer2} 
\end{align}
Note that $\tilde{\bm \beta}_{t}=(\tilde{\bm \beta}_{t,1},\ldots,\tilde{\bm \beta}_{t,p})$ can take a countable number of possible values. In addition, by construction (see Fig.~\ref{fig:noisy_communication_system}), 
the sequences $\{{\bm \alpha}_t:~t=0,1,\ldots\}$ and $\{\tilde{\bm \beta}_{t}:~t=0,1,\ldots\}$ are not Gaussian any more since by applying the change illustrated in Fig.~\ref{fig:replacement}, $\{{\bm \alpha}_t:~t=0,1,\ldots\}$ and $\{\tilde{\bm \beta}_{t}:~t=0,1,\ldots\}$ contain samples of the uniformly distributed process $\{{\bm r}_t:~t=0,1,\ldots\}$. As a result, the Kalman filter in Fig.~\ref{fig:noisy_communication_system} is no longer the least mean square estimator.

\paragraph*{Entropy coding} In what follows, we apply joint entropy coding across the vector dimension $p$ and memoryless coding across the time, that is, at each time step $t$ the output of the quantizer $\tilde{\bm \beta}_{t}$ is conditioned to the dither to generate a codeword ${\bf z}_t$. The decoder reproduces  ${\bm \beta}_{t}$ by subtracting the dithered signal ${\bf r}_t$ from $\tilde{\bm \beta}_{t}$. Specifically, at every time step $t$, we require that a message $\tilde{\bm \beta}_t$ is mapped into a codeword ${\bf z}_t\in\{0,1\}^{{\bf l}_t}$ designed using Shannon codes \cite[Chapter 5.4]{cover-thomas2006}. 
For a $\rv$ ${\bf x}$, the codes constructed based on Shannon coding scheme give an instantaneous (prefix-free) code with expected code length that satisfies the bounds 
\begin{align}
H({\bf x})\leq\mathbb{E}({\bf l})\leq{H}({\bf x})+1.\label{intantaneous_code_bounds}
\end{align}
If ${\bf x}$ is a $p-$dimensional random vector then the normalized version of \eqref{intantaneous_code_bounds} gives
\begin{align}
\frac{{H}({\bf x})}{p}\leq\frac{\mathbb{E}({\bf l})}{p}\leq\frac{{H}({\bf x})}{p}+\frac{1}{p}.\label{intantaneous_code_bounds:eq.1}
\end{align}

\begin{figure}[htp]
\centering
\includegraphics[width=0.7\columnwidth]{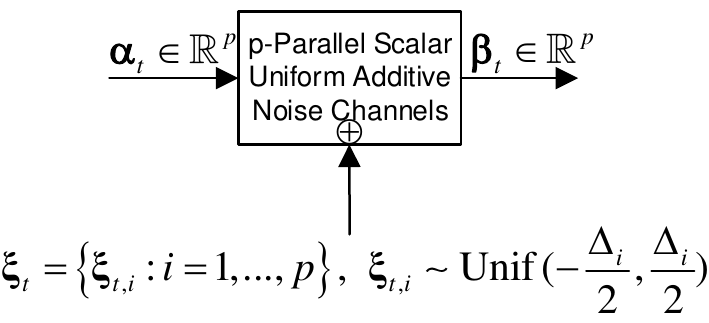}
\caption{An equivalent model to Fig.~\ref{fig:sdusq} based on scalar uniform additive noise channel\vspace{-0.3cm}.}
\label{fig:equivalent_additive_uniform_noise_channel}
\end{figure}

\par Since the $\sdusq$ operates using memoryless entropy coding over time, the following theorem holds.

\begin{theorem}(Upper bound)\label{theorem:general_bound}{\ \\}
Consider the realization of the zero-delay source-coding scheme illustrated in Fig.~\ref{fig:noisy_communication_system} with the change of $\awgn$ channel with $p-$parallel independently operating $\sdusq$ illustrated in Fig. \ref{fig:replacement}. If the vector process $\{\tilde{\bm  \beta}_t:~t=0,1,\ldots\}$ of the quantized output is jointly entropy coded conditioned to the dither signal values in a memoryless fashion for each $t$, then the operational Gaussian zero-delay rate, $R^{\op}_{\zd}(D)$, satisfies 
\begin{align}
R^{\op}_{\zd}(D)\leq{R}^{\na}(D)+\frac{p}{2}\log_2\left(\frac{\pi{e}}{6}\right)+1\label{operational_upper_bound:1}
\end{align}
where $p$ is the dimension of the state-space representation given in \eqref{state_space_model}, while the average $\mse$ distortion achieves the end-to-end average distortion $D$ of the system.
\end{theorem}
\begin{proof}
See Appendix~\ref{proof:theorem:general_bound}.
\end{proof}

\par The previous main result combined with the lower bound on Gaussian zero-delay $\rdf$, leads to the following corollary.
\begin{corollary}(Bounds on zero-delay $\rdf$)\label{corollary}{\ \\}
Consider the realization of the zero-delay source-coding scheme illustrated in Fig.~\ref{fig:noisy_communication_system} with the change of $\awgn$ channel with $p-$parallel independently operating $\sdusq$ as illustrated in Fig. \ref{fig:replacement}. Then, for vector (stable or unstable) Gaussian $\ar (1)$ sources the following bounds hold
\begin{align}
{R}^{\na}(D)\leq{R}^{\op}_{\zd}(D)\leq{R}^{\na}(D)+\frac{p}{2}\log\left(\frac{\pi{e}}{6}\right)+1.\label{general_bounds:eq.1}
\end{align}
\end{corollary}
\begin{proof}
This is obtained using the fact that ${R}^{\na}(D)\leq{R}^{\op}_{\zd}(D)$, \eqref{optimal_solution:eq.1} and Theorem \ref{theorem:general_bound}.
\end{proof}

\begin{theorem}(Generalization)\label{theorem:any_order}{\ \\}
The bounds derived in Corollary \ref{corollary} based on the realization scheme of Fig.~\ref{fig:noisy_communication_system} hold for vector Gaussian sources of any order.
\end{theorem}
\begin{proof}
See Appendix \ref{proof:theorem:any_order}.
\end{proof}

\par In the next remark, we draw connections to existing results in the literature.

\begin{remark}(Relations to existing results)\label{remark:discussion}{\ \\}
(1) For stationary stable scalar-valued Gaussian $\ar$ sources, our upper bound in Theorem \ref{theorem:general_bound} coincides with the bound obtained in \cite[Theorem 7]{derpich-ostergaard2012}. However, the upper bound in \cite{derpich-ostergaard2012} is obtained using a realization scheme with four filters instead of only one that we use in our scheme. In addition, our result takes into account unstable Gaussian sources too.\\
(2) Compare to \cite{silva-derpich-ostergaard2011}, we use $\ecdq$ based on a different realization setup that results into obtaining different lower and upper bounds. 
\end{remark} 

Next, we employ Theorem \ref{theorem:general_bound} to demonstrate a simulation example.

\begin{example}
 We consider a two-dimensional unstable Gaussian $\ar (1)$ source as follows:
\begin{align}
{\bf x}_{t+1} 
=\underbrace{\begin{bmatrix}
    -1.3 & 0.4 \\
     -0.3 & 0 
 \end{bmatrix}}_{A}
 {\bf x}_{t} 
+\underbrace{\begin{bmatrix}
    1 & 0 \\
    0 & 1 
 \end{bmatrix}}_{B} 
{\bf w}_{t}, 
  \label{nonstationary:two:dimension:equation1}
\end{align}
where ${\bf x}_{t}\in\mathbb{R}^2$, the parameter matrix $A$ is unstable because one of its eigenvalues, denoted by $\lambda_i(A)$, has magnitude greater than one, the pair ($A,B$) is stabilizable and ${\bf w}_{t}\sim{N}(0;I_{2\times{2}})$. By invoking SDPT3 \cite{tutuncu:2003} we plot the theoretical attainable lower and upper bounds to the zero-delay $\rdf$. This is illustrated in Fig. \ref{fig:unstable_rdf}. As expected from theory, $R^{\na}(D)\geq\sum_{\lambda_i(A)>1}\log|\lambda_i(A)|\approx{0.263}$~bits/source sample.
\begin{figure}[htp]
\centering
\includegraphics[width=\columnwidth]{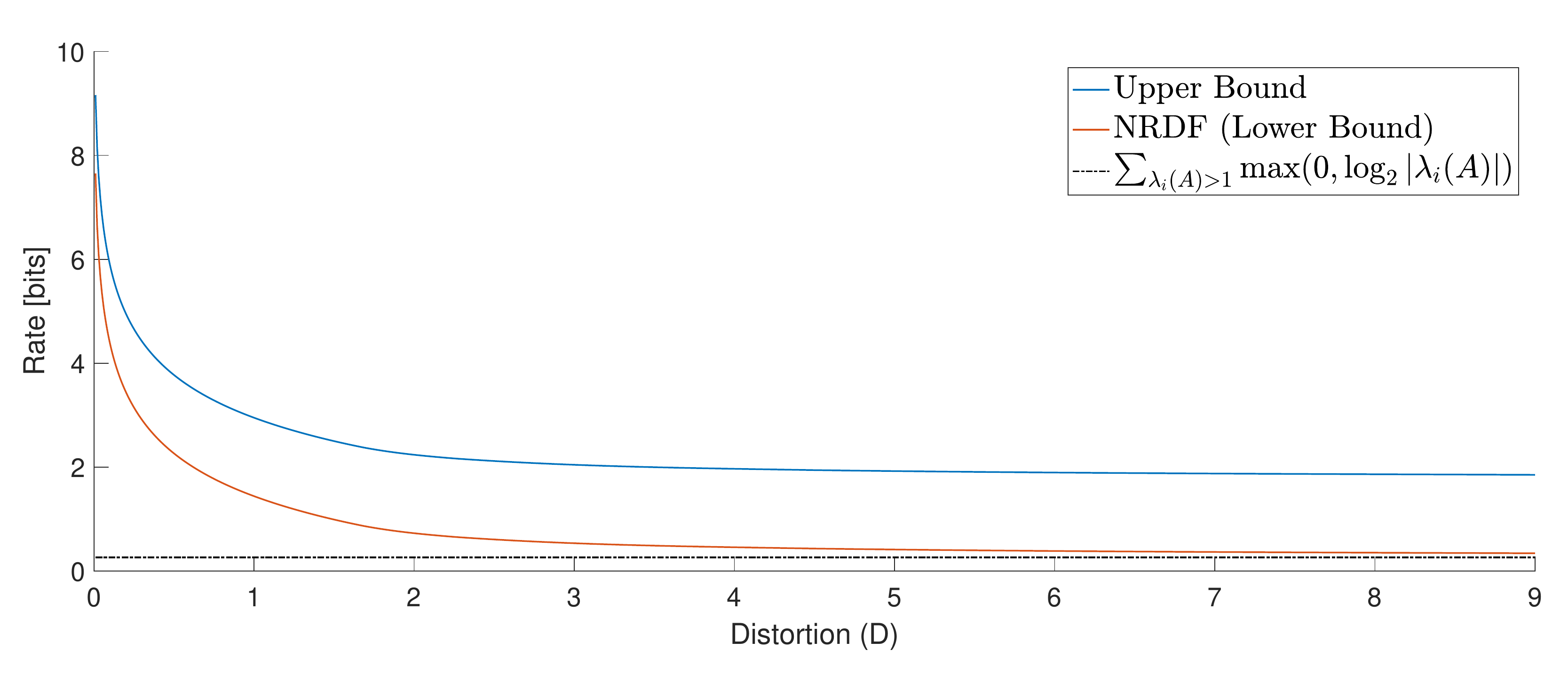}
\caption{Bounds on $R^{\op}_{\zd}(D)$ via the scheme of Fig. \ref{fig:noisy_communication_system}.}
\label{fig:unstable_rdf}
\end{figure}
\end{example}

\subsection{Vector Quantization}\label{subsection:vector_quantization}

\par It is interesting to observe that if instead of uniform scalar quantization we quantize over a lattice (vector) quantizer followed by memoryless entropy coded conditioned to the dither, then the upper bound in \eqref{operational_upper_bound:1} becomes
\begin{align}
R^{\op}_{\zd}(D)\leq{R}^{\na}(D)+\frac{p}{2}\log_2\left({2\pi{e}G_p}\right)+1\label{operational_upper_bound_vector_quantization:1}
\end{align}
where $G_p$ is the normalized second moment of the lattice \cite{zamir-feder1996}. If we take the average rate per dimension and assume an infinite dimensional vector Gaussian source, then by \cite[Lemma 1]{zamir-feder1996}, $G_p\rightarrow\frac{1}{2\pi{e}}$, and the terms due to space-filling loss and the loss due to entropy coding in \eqref{operational_upper_bound_vector_quantization:1} asymptotically goes to zero. Utilizing the latter, and the fact that $R^{\na}(D)\leq{R}_{\zd}^{\op}(D)$, we obtain
\begin{align}
\lim_{p\rightarrow\infty}\frac{1}{p}{R}^{\na}(D)\leq\lim_{p\rightarrow\infty}\frac{1}{p}{R}^{\op}_{\zd}(D)\leq\lim_{p\rightarrow\infty}\frac{1}{p}{R}^{\na}(D).\label{tight_bound}
\end{align}
i.e., ${R}^{\na}(D)$ is the $\opta$ by zero-delay codes.


%
%
%
%

\section{Conclusions and Future Directions}\label{sec:conclusions}
We considered zero-delay source coding of a vector Gaussian $\ar$ source under $\mse$ distortion. Based on a feedback realization scheme that quantizes the innovations of a Kalman filter with a $\sdusq$, we derived an upper bound to the zero-delay $\rdf$. We discussed the performance of this scheme when using lattice quantization. For infinite dimensions we observed that the $\nrdf$ coincide with the zero-delay $\rdf$. An illustrative example is presented to support our findings. 
\par As an ongoing research, we will apply the proposed coding scheme based on $\sdusq$ to find the actual operational rates corresponding to the zero-delay $\rdf$. Moreover, we will examine similar coding schemes for fixed-length coding rate.

%
%
%
%
%
%

\appendices

\section{Proof of Theorem~\ref{theorem:general_bound}}\label{proof:theorem:general_bound}

\par In the realization scheme proposed in Fig.~\ref{fig:noisy_communication_system}, with the change of $\awgn$ channel with $p-$parallel independently operating $\sdusq$, the operational rate for each $t$ is equal to the conditional entropy $H(\tilde{\bm \beta}_t|{\bf r}_t)$ where $\tilde{\bm \beta}_{t}=\{\tilde{\bm \beta}_{t,1},\ldots,\tilde{\bm \beta}_{t,p}\}$, $\tilde{\bm \beta}_{t,i}=Q_{\Delta_i}({\bm \alpha}_{t,i}+{\bf r}_{t,i}),~i=1,\ldots,p$, i.e., the entropy of the quantized output $\tilde{\bm \beta}_t$ conditioned on the $t-$value of the dither signal ${\bm r}_t$. This leads to the following analysis.
\begin{align}
&H(\tilde{\bm \beta}_t|{\bf r}_t)\stackrel{(a)}=I({\bm \alpha}_{t};{\bm \beta}_{t})\nonumber\\
&\stackrel{(b)}=I({\bm \alpha}_{t};{\bm \alpha}_{t}+{\bm \xi}_{t})\nonumber\\
&={h}({\bm \alpha}_{t}+{\bm \xi}_{t})-h({\bm \xi}_{t})\nonumber\\
&\stackrel{(c)}={h}({\bm \alpha}^G_{t}+{\bf v}_{t})-h({\bm \alpha}^G_{t})+\mathbb{D}({\bm \xi}_{t}||{\bf v}_t)-\mathbb{D}({\bm \alpha}_{t}+{\bm \xi}_{t}||{\bm \alpha}^G_{t}+{\bf v}_{t})\nonumber\\
&\stackrel{(d)}\leq{I}({\bm \alpha}^G_{t};{\bm \alpha}^G_{t}+{\bf v}_{t})+\mathbb{D}({\bm \xi}_{t}||{\bf v}_t)\nonumber\\
&\stackrel{(e)}=I({\bm \alpha}^G_{t};{\bm \alpha}^G_{t}+{\bf v}_{t})+\frac{p}{2}\log_2\left(\frac{\pi{e}}{6}\right)\nonumber\\
&=I({\bm \alpha}^G_{t};{\bm \beta}^G_{t})+\frac{p}{2}\log_2\left(\frac{\pi{e}}{6}\right),\label{proof:eq.1}
\end{align}
where $(a)$ follows from \cite[Theorem 1]{zamir-feder1992}; $(b)$ follows from the fact that the quantization noise is ${\bm \xi}_t={\bm \beta}_t-{\bm \alpha}_t$ (see Fig.~\ref{fig:equivalent_additive_uniform_noise_channel}); $(c)$ follows from the fact that the relative entropy $\mathbb{D}(x||x^\prime)=h(x^\prime)-h(x)$, see, e.g., \cite[Theorem 8.6.5]{cover-thomas2006}; $(d)$ follows from the fact that $D({\bm \alpha}_{t}+{\bm \xi}_{t}||{\bm \alpha}^G_{t}+{\bf v}_{t})\geq{0}$, with equality if and only if $\{{\bm \xi}_t:~t=0,1,\ldots\}$ becomes a Gaussian distribution; $(e)$ from the fact that the differential entropy $h({\bf v}_t)$ of a Gaussian random vector with covariance $\Sigma_{\bf v}\triangleq\diag\{V\}$ is 
\begin{align*}
h({\bf v}_t)=\frac{1}{2}\log_2(2\pi{e})^p|\Sigma_{\bf v}|=\sum_{i=1}^p\frac{1}{2}\log_2(2\pi{e})V_{ii},
\end{align*}
and the entropy $h({\bm \xi}_t)$ of the uniformly distributed random vector ${\bm \xi}_t=\{{\bm \xi}_{t,i}:~i=1,2,\ldots,p\}, {\bm \xi}_{t,i}\sim\unif\left(-\frac{\Delta_i}{2},\frac{\Delta_i}{2}\right)$ is
\begin{align*}
h({\bm \xi}_t)=\sum_{i=1}^p\frac{1}{2}\log_2\Delta^2_i.
\end{align*}
Since we have that $V_{ii}=\frac{\Delta^2_i}{12}$, $i=1,\ldots,p$, the result follows.
\par Next, note that for $n=0,1,\ldots$, the following inequality holds in Fig.~\ref{fig:noisy_communication_system}.
\begin{align}
{I}({\bf x}^n;{\bf y}^n)&\stackrel{(a)}=\sum_{t=0}^n{I}({\bf x}_t;{\bf y}_t|{\bf y}^{t-1})\nonumber\\
&\stackrel{(b)}=\sum_{t=0}^n{I}({\bf k}_t;\tilde{\bf k}_t)\nonumber\\
&\stackrel{(c)}=\sum_{t=0}^n{I}({\bm \alpha}_t;{\bm \beta}_t),\label{proof:eq.2}
\end{align}
where $(a)$ follows from the structural properties of specific extremum problem resulting in the realization of Fig. \ref{fig:noisy_communication_system} (see, e.g., the analysis in Section~\ref{sec:problem_formulation} and \cite[Remark IV.5]{charalambous-stavrou-ahmed2014}); $(b)$ follows from the analysis in \cite[Equation (35)]{charalambous-stavrou-ahmed2014}; $(c)$ follows from the fact that $E, \Phi, \Theta$ are invertible matrices and as a result the information from ${\bf k}_t$ to $\tilde{\bf k}_t$ is the same as from ${\bm \alpha}_t$ to ${\bm \beta}_t$ (information lossless operation).  

\par Since we are assuming joint memoryless entropy coding of $p$ independently operating scalar uniform quantizers with subtractive dither, then by \eqref{intantaneous_code_bounds}, for $t=0,1,\ldots,n$, we obtain
\begin{align}
\sum_{t=0}^n\mathbb{E}({\bf l}_t)&\leq\sum_{t=0}^n\left({H}(\tilde{\bm \beta}_t|{\bf r}_t)+1\right)\nonumber\\
&\stackrel{(a)}\leq\sum_{t=0}^n\left(I({\bm \alpha}^G_{t};{\bm \beta}^G_{t})+\frac{p}{2}\log_2\left(\frac{\pi{e}}{6}\right)+1\right)\nonumber\\
&\stackrel{(b)}\leq{I}({\bf x}^{n,G};{\bf y}^{n,G})+(n+1)\frac{p}{2}\log_2\left(\frac{\pi{e}}{6}\right)+(n+1),\label{proof:eq.3}
\end{align}
where $(a)$ follows by \eqref{proof:eq.1} and $(b)$ follows from \eqref{proof:eq.2}.

\par Then, by first taking the per unit time limiting expression in \eqref{proof:eq.3} and then the infimum, we obtain 
\begin{align}
&\inf\limsup_{n\longrightarrow\infty}\frac{1}{n+1}\sum_{t=0}^n\mathbb{E}({\bf l}_t)\nonumber\\
&\leq\inf\limsup_{n\longrightarrow\infty}\frac{1}{n+1}{I}({\bf x}^{n,G};{\bf y}^{n,G})+\frac{p}{2}\log_2\left(\frac{\pi{e}}{6}\right)+1\nonumber\\
&\stackrel{(a)}\Longrightarrow{R}_{\zd}^{\op}\leq\bar{R}^{\na}(D)+\frac{p}{2}\log_2\left(\frac{\pi{e}}{6}\right)+1,\label{proof:eq.4}
\end{align}
where $(a)$ follows by \eqref{def:operational_zero_delay} and \eqref{infinite_time_nrdf_upper_bound} respectively, and $\bar{R}^{\na}(D)$ is the upper bound expression of $R^{\na}(D)$ for the vector Gaussian $\ar (1)$ source model given by \eqref{state_space_model}.
\par Finally, by assumptions, (i.e., $A$ is stable or $(A,B)$ stabilizable) the innovations' Gaussian source is asymptotically stationary. Since $R^{\na}(D)=R^{\na,{{\bf k}^n},\tilde{\bf k}^n}(D)$, then at steady state, we have $\bar{R}^{\na}(D)=R^{\na}(D)$ and the result follows.\qed

\section{Proof of Theorem~\ref{theorem:any_order}}\label{proof:theorem:any_order}

Assume the following vector Gaussian $\ar(s)$ process, where $s$ is a positive integer, in state space representation.
\begin{align}
{\bf x}_{t+1}=\sum_{j=1}^sA_j{\bf x}_{t-j+1}+B{\bf w}_t,\label{vector_ar-p}
\end{align}
where $A_j\in\mathbb{R}^{p\times{p}}$, and $B\in\mathbb{R}^{p\times{q}}$ are deterministic matrices, ${\bf x}_0\in\mathbb{R}^p\sim{N}(0;\Sigma_{{\bf x}_0})$ is the initial state, and ${\bf w}_t\in\mathbb{R}^q\sim{N}(0;I_{q\times{q}})$ is an $\IID$ Gaussian sequence, independent of ${\bf x}_0$. Clearly, for $s=1$, \eqref{vector_ar-p} gives as a special case the source model described by \eqref{state_space_model}.\\
Next, we show that \eqref{vector_ar-p} can be expressed as an augmented vector Gaussian $\ar (1)$ process as follows
\begin{align}
\tilde{\bf x}_{t+1}=\tilde{A}\tilde{\bf x}_{t}+\tilde{B}\tilde{\bf w}_t,\label{vector_ar1}
\end{align}
where $\tilde{A}\in\mathbb{R}^{sp\times{sp}}$, and $\tilde{B}\in\mathbb{R}^{sp\times{sq}}$ are deterministic matrices, $\tilde{\bf x}_0\in\mathbb{R}^{sp}\sim{N}(0;\Sigma_{\tilde{\bf x}_0})$ is the initial state with $\Sigma_{\tilde{\bf x}_0}$ being the covariance of the initial state, and $\tilde{\bf w}_t\in\mathbb{R}^{sq}\sim{N}(0;\Sigma_{\tilde{\bf w}_t})$ is an $\IID$ Gaussian sequence, independent of $\tilde{\bf x}_0$.

The proof employs a simple augmentation of the state process. First, note that the state space model of \eqref{vector_ar-p} can be modified as follows.
\begin{align}
\begin{bmatrix}
     {\bf x}_{t+1} \\ 
     {\bf x}_{t}\\
     \vdots\\
     {\bf x}_{t-s+2}
\end{bmatrix}&=
\begin{bmatrix}
      A_1 & {A}_2 &\ldots & A_{s-1}&A_s\\ 
      I & 0 & \ldots &0 & 0\\
      0& I & \ddots &0 & 0\\
      \vdots & \vdots&\ddots & \vdots & \vdots\\
      0 & 0 & \ldots &I & 0\\
\end{bmatrix}
\begin{bmatrix}
     {\bf x}_{t} \\ 
     {\bf x}_{t-1}\\
     \vdots\\
     {\bf x}_{t-s+1}
\end{bmatrix}\nonumber\\
&+\begin{bmatrix}
      B & 0 &\ldots & 0\\ 
      0 & 0 & \ldots & 0\\
     \vdots & \vdots & \ldots & \vdots\\
      0 & 0 & \ldots & 0\\
\end{bmatrix}
\begin{bmatrix}
     {\bf w}_{t} \\ 
     0\\
     \vdots\\
     0
\end{bmatrix}.\label{proof:aug.eq.3}
\end{align}
Then, \eqref{proof:aug.eq.3} can be written in an augmented state space form as follows.
\begin{align}
\tilde{\bf x}_{t+1}=\tilde{A}\tilde{\bf x}_t+{B}\tilde{\bf w}_t, \label{proof:aug.eq.4}
\end{align}
where
\begin{align*}
\tilde{\bf x}_{t+1}&=\begin{bmatrix}
     {\bf x}_{t+1} \\ 
     {\bf x}_{t}\\
     \vdots\\
     {\bf x}_{t-s+2}
\end{bmatrix}\in\mathbb{R}^{sp},
\tilde{\bf x}_t=\begin{bmatrix}
     {\bf x}_{t} \\ 
     {\bf x}_{t-1}\\
     \vdots\\
     {\bf x}_{t-s+1}
\end{bmatrix}\in\mathbb{R}^{sp},\\
\tilde{A}&=\begin{bmatrix}
      A_1 & {A}_2 &\ldots & A_{s-1}&A_s\\ 
      I & 0 & \ldots &0 & 0\\
      0& I & \ddots &0 & 0\\
      \vdots & \vdots&\ddots & \vdots & \vdots\\
      0 & 0 & \ldots &I & 0\\
\end{bmatrix}\in\mathbb{R}^{sp\times{sp}},\\
~\tilde{B}&=\begin{bmatrix}
      B & 0 &\ldots & 0\\ 
      0 & 0 & \ldots & 0\\
     \vdots & \vdots & \ldots & \vdots\\
      0 & 0 & \ldots & 0\\
\end{bmatrix}\in\mathbb{R}^{sp\times{sq}},
~\tilde{\bf w}_t=\begin{bmatrix}
     {\bf w}_{t} \\ 
     0\\
     \vdots\\
     0
\end{bmatrix}\in\mathbb{R}^{sq}.
\end{align*}
The augmented state space representation of \eqref{proof:aug.eq.4} is an augmented vector Gaussian $\ar (1)$ process which can then be applied to the feedback design of Fig.~\ref{fig:noisy_communication_system} obtaining the same bound as in Theorem~\ref{theorem:general_bound}. This completes the proof.\qed

\bibliographystyle{IEEEtran}
\bibliography{string,photis_quantization}

\end{document}